\documentclass[12pt]{article}

\usepackage{fullpage}
\usepackage[T1]{fontenc}
\usepackage{ae,aecompl}
\usepackage[dvips]{graphicx}
\usepackage{amssymb}
\usepackage{amsmath}
\usepackage{subfigure}
\usepackage[round,authoryear]{natbib}
\usepackage{color}
\usepackage{array}
\usepackage{rotating}
\usepackage{colortbl}
\usepackage{tikz}
\usetikzlibrary{patterns}

\definecolor{webgreen}{rgb}{0,0.4,0}
\definecolor{webbrown}{rgb}{0.6,0,0}
\definecolor{purple}{rgb}{0.5,0,0.25}
\definecolor{darkblue}{rgb}{0,0,0.7}
\definecolor{darkred}{rgb}{0.7,0,0}
\definecolor{darkgreen}{rgb}{0,0.7,0}
\usepackage[pdfborder=false]{hyperref}
\hypersetup{colorlinks,citecolor=darkred,filecolor=black,linkcolor=darkblue,urlcolor=webgreen,pdfpagemode=none,
pdfstartview=FitH}
\newcommand{\ignore}[1]{}
\newtheorem{lemma}{{\sc Lemma}}
\newtheorem{prop}{{\sc Proposition}}
\newtheorem{cor}{{\sc Corollary}}
\newtheorem{theorem}{{\sc Theorem}}
\newtheorem{defn}{{\sc Definition}}

\newenvironment{proof}{\noindent {\bf \sl Proof\/}:\enspace}
{\hfill $\blacksquare{}$ \vspace{12pt}}
\usepackage{sectsty}
\sectionfont{\fontsize{14}{14}\centering\normalfont\scshape}
\subsectionfont{\centering\normalfont}

\begin{document}

\title{{\bf {\Large Ordinal Bayesian incentive compatibility \\
in random assignment model}}~\thanks{We are grateful to Sven Seuken, Timo Mennle, Arunava Sen, Dipjyoti Majumdar, Souvik Roy, and Wonki Cho for
their comments.}}
\date{\today}
\author{Sulagna Dasgupta and Debasis Mishra~\thanks{Dasgupta: University of Chicago,~\texttt{sulagna@uchicago.edu}; Mishra: Indian Statistical Institute, Delhi, \texttt{dmishra@isid.ac.in}}}

\maketitle

\begin{abstract}

We explore the consequences of weakening the
notion of incentive compatibility from strategy-proofness to ordinal Bayesian incentive compatibility (OBIC)
in the random assignment model. If the common prior of the agents
is the {\em uniform prior}, then a large class of random mechanisms are OBIC with respect
to this prior -- this includes the probabilistic serial mechanism.
We then introduce a robust version of OBIC: a mechanism is {\em locally robust OBIC} if
it is OBIC with respect {\em all}
independent and identical priors in some neighborhood of a given independent and identical prior.
We show that every locally robust OBIC mechanism
satisfying a mild property called {\em elementary monotonicity} is strategy-proof. This leads to a
strengthening of the impossibility result in \citet{Bogo01}: if there are at least four agents, there is
no locally robust OBIC and ordinally efficient mechanism satisfying equal treatment of equals. \\


\end{abstract}

\newpage

\section{Introduction}

This paper explores the consequences of weakening incentive compatibility from strategy-proofness to {\em ordinal Bayesian incentive compatibility} in
the random assignment model (one-sided matching model). Ordinal Bayesian incentive compatibility (OBIC) requires
that the truth-telling {\em expected share vector} of an agent first-order stochastically dominates
the expected share vector from reporting any other preference. It is the natural analogue of Bayesian incentive compatibility
in an ordinal mechanism. This weakening of
strategy-proofness was proposed by \citet{dAsp88}.
We study OBIC by considering mechanisms that allow for randomization in the assignment model.

In the random assignment model, the set of mechanisms satisfying ex-post efficiency and
strategy-proofness is quite rich.\footnote{\citet{Pycia17} characterize the
set of deterministic, strategy-proof, Pareto efficient, and non-bossy mechanisms
in this model. This includes generalizations of the top-trading-cycle
mechanism.} Despite satisfying such strong incentive properties,
all of them either fail to satisfy {\em equal treatment of equals}, a weak notion of fairness, or
{\em ordinal efficiency}.
Indeed, \citet{Bogo01} propose a new mechanism, called the {\em probabilistic
serial mechanism}, which satisfies equal treatment of equals and ordinal efficiency.
However, they show that it fails strategy-proofness, and no mechanism can satisfy all these
three properties simultaneously if there are at least four agents.
A primary motivation for weakening the notion of incentive compatibility to OBIC
is to investigate if we can escape this impossibility result.

We show two types of results.
First, if the (common) prior is a uniform probability distribution over the set of possible
preferences, then every {\em neutral} mechanism satisfying a mild property called {\em elementary
monotonicity} is OBIC.\footnote{Neutrality is a standard axiom in social choice theory
which requires that objects are treated symmetrically. Elementary monotonicity is
a monotonicity requirement of a mechanism. We define it formally in Section \ref{sec:uniform}.} An example of such a mechanism is the probabilistic serial mechanism.
This is a positive result and provides a strategic foundation for the probabilistic
serial mechanism. In particular, it shows that there exist ordinally efficient mechanisms
satisfying equal treatment of equals which are OBIC with respect to the uniform prior.

Second, we explore the implications of strengthening OBIC as follows. A mechanism is {\em locally robust OBIC} (LROBIC) with
respect to an independent and identical prior if it is OBIC with respect to every independent and identical prior in
its ``neighborhood". The motivation for such requirement of robustness in the
mechanism design literature is now well-known, and referred to as the Wilson doctrine \citep{Wilson87}.
We show that every LROBIC mechanism satisfying elementary monotonicity is
strategy-proof. An immediate corollary of this result is that the probabilistic serial mechanism
is not LROBIC (though it is OBIC with respect to the uniform prior). As a corollary, we can show that
when there are at least four agents, there is no LROBIC and {\em ordinally efficient} mechanism satisfying
equal treatment of equals. This strengthens the seminal impossibility result of \citet{Bogo01} by
replacing strategy-proofness with LROBIC.

Both our results point to very different implications of OBIC in the presence of
elementary monotonicity -- if the prior is uniform, this notion of incentive compatibility
is very permissive; but if we require OBIC with respect to a set of independent and identical priors in any
neighborhood of a given prior, this notion of incentive compatibility is very restrictive.

\subsection{Related literature}

There is fairly large literature on random assignment problems. We
summarize them below.

The notion of incentive compatibility that we use, OBIC, has been used
in voting models by \citet{Majumdar04,Bhargava15,Mishra16,Kim18} to escape the
dictatorship result in \citep{Gibbard73,Satt75,Gibbard77}. All these papers use deterministic mechanisms
in voting models, whereas we apply OBIC to the random assignment model.
\citet{Majumdar04} show that every deterministic
neutral voting mechanism satisfying elementary monotonicity
is OBIC with respect to uniform priors. Our Theorem \ref{theo:uniform} shows
that this result generalizes to the random assignment model.
\citet{Mishra16} generalizes this result to some restricted domains of voting (like the single peaked domain).
He shows that in the deterministic voting model, elementary monotonicity and
OBIC with respect to ``generic" prior is equivalent to strategy-proofness in
a variety of restricted domains -- see also \citet{Kim18} for a strengthening
of this result. Though these results are similar to our Theorem \ref{theo:lobic},
there are significant differences. First, we consider randomization while these
results are only for deterministic mechanisms. Our notion of locally robust OBIC
is incomparable to OBIC with respect to generic priors used in these papers.
Second, ours is a model of private good allocation (random assignment), while
these papers deal with the voting model.

\citet{Bogo01}
introduce a family of mechanisms in the random assignment model.
They call these the simultaneous eating algorithms. which
generate {\sl ordinally} efficient random assignments, a stronger notion of efficiency than ex-post efficiency.
\footnote{\citet{Katta06} extend
the simultaneous eating algorithm to allow for ties in preferences.}
The probabilistic serial mechanism belongs to this family and it is anonymous.
However, it is not strategy-proof. In fact, \citet{Bogo01} show
that there is no ordinally efficient and strategy-proof mechanism satisfying
equal treatment of equals when there are at least four agents.\footnote{With three
agent, the random priority mechanism satisfies these properties.}

There is a large literature that provides strategic foundations
to the probabilistic serial (PS) mechanism. \citet{Bogo01} show that
the PS mechanism satisfies {\em weak}-strategy-proofness. Their notion
of weak strategy-proofness requires that the manipulation
share vector {\em cannot} first-order-stochastic-dominate the truth-telling share vector.
\citet{Bogo02} study a problem where agents have an {\em outside option}.
When agents have the same ordinal ranking over objects but the position of
outside option in the ranking of objects is the only private information, they show that the PS mechanism is
strategy-proof. Other contributions in this direction include \citet{Liu19,Liu19b}, who identify
domains where the probabilistic serial mechanism is strategy-proof.
\citet{Che10} show that the PS mechanism and the random priority mechanism (which
is strategy-proof) are asymptotically equivalent. Similarly, \citet{Kojima10} show
that when sufficiently many copies of an object are present, then the PS mechanism
is strategy-proof. Thus, in large economies, the PS mechanism is strategy-proof.
\cite{Balbu16} introduce a notion of strategy-proofness which is stronger than
weak strategy-proofness and show that the PS mechanism satisfies it.
His notion of strategy-proofness is based on the ``convex" domination of lotteries, and
hence, called {\em convex strategy-proofness}. \citet{Me21}
define a notion called {\em partial strategy-proofness}, which is weaker
than strategy-proofness and show that the PS mechanism satisfies it.
They show that strategy-proofness is equivalent to upper invariant, lower
invariant and elementary monotonicity (they call it swap monotonicity).
Their notion of partial strategy-proofness is equivalent to upper invariance
and elementary monotonicity, and hence, it is weaker than strategy-proofness.

The main difference between
these weakenings of strategy-proofness and ours is that OBIC is a {\em prior-based}
notion of incentive compatibility. It is the natural analogue of Bayesian
incentive compatibility in an ordinal environment. \citet{Ehlers07} study
OBIC in a {\em two-sided matching} problem. Their main focus is on
OBIC mechanism that select a stable matching. They characterize the beliefs for which such a
mechanism exists. There is a literature in computer science studying computational aspects of manipulation
of the PS rule -- see \citet{Az14,Az15} and references therein.

\section{Model}

{\bf Assignments.} There are $n$ agents and $n$ objects.\footnote{All our results
extend even if the number of objects is not the same as the number of agents. We assume this
only to compare our results with the random assignment literature, where this assumption
is common.}
Let $N:=\{1,\ldots,n\}$ be the set of
agents and $A$ be the set of objects.
We define the notion of a feasible assignment
first.

\begin{defn}
An $n \times n$ matrix $L$ is an {\bf assignment}
if
\begin{align*}
L_{ia} &\in [0,1]~\qquad~\forall~i \in N,~\forall~a \in A \\
\sum_{a \in A}L_{ia} &= 1 ~\qquad~\forall~i \in N\\
\sum_{i \in N}L_{ia} &= 1~\qquad~\forall~a \in A
\end{align*}
\end{defn}
Hence, an assignment is a bistochastic matrix.
For any assignment $L$, we write $L_i$ as the {\bf share vector} of
agent $i$.\footnote{Whenever we say an assignment, we mean a random assignment from now on.} Formally, a share vector is a probability distribution
over the set of objects. For any $i \in N$ and any $a \in A$, $L_{ia}$ denotes the ``share" of
agent $i$ of object $a$. The second constraint of the assignment
definition requires that the total share of every agent is $1$.
The third constraint of the assignment requires that every object is
completely assigned.
Let $\mathcal{L}$ be the set of all assignments.

An assignment $L$ is {\bf deterministic} if $L_{ia} \in \{0,1\}$ for all
$i \in N$ and for all $a \in A$. Let $\mathcal{L}^d$ be the set of all
deterministic assignments. By the Birkohff-von-Neumann theorem,
for every $L \in \mathcal{L}$, there exists a set of deterministic assignments
in $\mathcal{L}^d$ whose convex combination equals $L$.

{\bf Preferences.}
A preference is a strict ordering of $A$. The preference of an agent $i$ will be
denoted by $P_i$.
The set of all preferences over $A$ is denoted by $\mathcal{P}$.
A preference profile
is $\mathbf{P} \equiv (P_1,\ldots,P_{n})$, and we will denote by
$P_{-i}$ the preference profile $\mathbf{P}$ excluding the preference
$P_i$ of agent $i$. We write $aP_ib$ to denote that $a$ is strictly preferred over $b$
in preference $P_i$.

{\bf Prior.} We assume that the preference of each agent
is independently and identically drawn using a common prior
$\mu$, which is a probability distribution over $\mathcal{P}$.
From now on, whenever we say a prior, we refer to such an independent and identical prior.
We will denote by $\mu(P_i)$ the probability with which
agent $i$ has preference $P_i$. With some abuse of notation, we will
denote the probability with which agents in $N \setminus \{i\}$ have
preference profile $P_{-i}$ as $\mu(P_{-i})$. Note that by independence,
$\mu(P_{-i}) = \times_{j \ne i}\mu(P_j)$.

\section{Ordinal Bayesian incentive compatibility}

Our solution concept is Bayes-Nash equilibrium but
we restrict attention to ordinal mechanisms, i.e., mechanisms where
we only elicit ranking over objects from each agent.
Hence, whenever we say {\em mechanism}, we refer to such ordinal mechanisms.\footnote{
The restriction to not consider cardinal mechanisms is arguably arbitrary. It is usually done to
simplify the process of elicitation. Such restriction is also consistent
with the literature on random assignment models. The set of
incentive compatible mechanisms expand if we consider cardinal mechanisms~\citep{Mi12,Ab20}.}
Formally, a {\bf mechanism} is a map $Q: \mathcal{P}^n \rightarrow \mathcal{L}$.
A mechanism $Q$
assigns a share vector $Q_i(\mathbf{P})$ to agent $i$ at every preference profile $\mathbf{P}$.

Before discussing the notions of incentive compatibility, it is useful to think how
agents compare share vectors in our model. Fix agent $i$ with a preference $P_i$
over the set of objects $A$. Denote the $k$-th ranked object in $P_i$ as $P_i(k)$.
Consider two share vectors $\pi, \pi'$.
For every $a \in A$, we will denote by $\pi_a$ and $\pi'_a$ the share
assigned to object $a$ in $\pi$ and $\pi'$ respectively.
We will say $\pi$ {\bf first-order-stochastically-dominates (FOSD)} $\pi'$ according to
$P_i$ if
\begin{align*}
\sum_{k=1}^{\ell} \pi_{P_i(k)} \ge \sum_{k=1}^{\ell}\pi'_{P_i(k)}~\qquad~\forall~\ell \in \{1,\ldots,n\}.
\end{align*}
In this case, we will write $\pi \succ_{P_i} \pi'$. Notice that $\succ_{P_i}$ is not
a complete relation over the outcomes. An equivalent (and well known) definition
of $\succ_{P_i}$ relation is that for {\em every} von-Neumann-Morgenstern utility representation of
$P_i$, the expected utility from $\pi$ is at least as much as $\pi'$.

The most standard notion of incentive compatibility is
strategy-proofness (dominant strategy incentive compatibility), which uses the FOSD
relation to compare share vectors.
\begin{defn}
A mechanism $Q$ is {\bf strategy-proof} if for every $i \in N$, every $P_{-i} \in \mathcal{P}^{n-1}$,
and every $P_i, P'_i \in \mathcal{P}$, we have
\begin{align*}
Q_i(P_i,P_{-i}) \succ_{P_i} Q_i(P'_i,P_{-i}).
\end{align*}
\end{defn}
The interpretation of this definition is that fixing the preferences of other agents, the
truth-telling share vector must FOSD other share vectors that can be obtained by deviation.
This definition of strategy-proofness appeared in \citet{Gibbard77} for voting problems, and has been the
standard notion in the literature on random voting and random assignment problems.

The ordinal Bayesian incentive compatibility notion is an adaptation of this by changing the
solution concept to Bayes-Nash equilibrium. It was first introduced and studied in a voting committee
model in \citet{dAsp88}, and was later used in many voting models~\citep{Majumdar04}.
To define it formally, we introduce the notion of an interim share vector.
Fix an agent $i$ with preference $P_i$.
Given a mechanism $Q$, the {\bf interim share} of object $a$ for agent $i$ by reporting $P'_i$
is:
\begin{align*}
q_{ia}(P'_i) &= \sum_{P_{-i} \in \mathcal{P}^{n-1}}\mu(P_{-i})Q_{ia}(P'_i,P_{-i}).
\end{align*}
The {\bf interim share vector} of agent $i$ by reporting $P'_i$ will be denoted
as $q_{i}(P'_i)$.
\begin{defn}
  A mechanism $Q$ is {\bf ordinally Bayesian incentive compatible (OBIC)} (with respect to prior $\mu$)
  if for every $i \in N$ and every $P_i, P'_i \in \mathcal{P}$, we have
  \begin{align*}
  q_i(P_i) \succ_{P_i} q_i(P'_i).
  \end{align*}
\end{defn}
It is immediate that if a mechanism $Q$ is strategy-proof it is OBIC with respect to every (including correlated and non-identical) prior.
Conversely, if a mechanism is OBIC with respect to {\em all} priors (including correlated and non-identical priors),
then it is strategy-proof.

\subsection{A motivating example}

We investigate a simple example to understand the implications of strategy-proofness and OBIC
for the probabilistic serial mechanism.
Suppose $n=3$ with three objects $\{a,b,c\}$. Consider the preference profiles $(P_1,P_2,P_3)$
and $(P'_1,P_2,P_3)$ shown in Table \ref{tab:t1} -- the table also shows the share vector
of each agent in the probabilistic serial mechanism of \citet{Bogo01}. In the probabilistic serial
mechanism, each agent starts ``eating" her favorite object simultaneously till the object
is finished. Then, she moves to the best available object according to her preference and so on.
Each agent has the same eating speed.
Table \ref{tab:t1} shows the output of the probabilistic serial mechanism for
preference profiles $(P_1,P_2,P_3)$ and $(P'_1,P_2,P_3)$. Since $Q_{1a}(P'_1,P_2,P_3)+Q_{1c}(P'_1,P_2,P_3)
> Q_{1a}(P_1,P_2,P_3)+Q_{1c}(P_1,P_2,P_3)$, we conclude that $Q_1(P_1,P_2,P_3) \nsucc_{P'_1} Q_1(P'_1,P_2,P_3)$.
Hence, agent $1$ can manipulate from $P_1$ to $P'_1$, when agents $2$ and $3$ have
preferences $(P_2,P_3)$.

\begin{table}
\centering
\begin{tabular}{|c|c|c||c|c|c|}
\hline
$P_1$ & $P_2$ & $P_3$ & $P'_1$ & $P_2$ & $P_3$ \\
\hline
\hline
$c$; $\frac{1}{2}$ & $a$; $\frac{1}{2}$ & $c$; $\frac{3}{4}$ & $a$; $\frac{1}{2}$ & $a$; $\frac{2}{3}$ & $c$; $\frac{1}{2}$ \\
$a$; $\frac{1}{6}$ & $b$; $\frac{1}{2}$ & $a$; $0$ & $c$; $\frac{1}{4}$ & $b$; $\frac{1}{3}$ & $a$; $\frac{1}{6}$ \\
$b$; $\frac{1}{3}$ & $c$; $0$ & $b$; $\frac{1}{4}$ & $b$; $\frac{1}{4}$ & $c$; $0$ & $b$; $\frac{1}{3}$ \\
\hline
\end{tabular}
\label{tab:t1}
\caption{Manipulation by agent $1$.}
\end{table}

When can such a manipulation be prevented by OBIC?
Note that $P_1$ is generated from $P'_1$ by permuting $a$ and $c$.
Suppose we permute $P_2$ and $P_3$ also to get $P'_2$ and $P'_3$ respectively:
$$c~P'_2~b~P'_2~a~\textrm{and}~a~P'_3~c~P'_3~b.$$
Since the probabilistic serial mechanism is {\em neutral} (with respect to objects),
the share vector of agent $1$ at $(P_1,P_2,P_3)$ is a permutation of its share vector at
$(P'_1,P'_2,P'_3)$.
Further, when all the preferences are equally likely,
the probability of $(P_2,P_3)$ is equal to the probability of $(P'_2,P'_3)$.
So, the total expected probability of $a$ and $c$ for agent $1$ at $P_1$ and $P'_1$
is the same (where expectation is taken over $(P_2,P_3)$ and $(P'_2,P'_3)$).
As
we show below, this argument generalizes and the expected share vector at $P_1$
first-order-stochastic-dominates the expected share vector at $P'_1$ when the true preference is $P_1$
and prior is uniform.

\section{Uniform prior and possibilities}
\label{sec:uniform}

In this section, we present our first result which shows that the set of OBIC mechanisms
is much larger than the set of strategy-proof mechanisms if the prior is the uniform prior.
A prior $\mu$ is the {\bf uniform prior} if
$\mu(P_i) = \frac{1}{|\mathcal{P}|} = \frac{1}{n!}$ for each $P_i \in \mathcal{P}$.
Uniform prior puts equal probability on each of the possible preferences.
We call a mechanism {\bf U-OBIC} if it is OBIC with respect to the uniform prior.

We show that there is a large class of mechanisms which are U-OBIC - this will include some well-known
mechanisms which are known to be not strategy-proof. This class is characterized by two axioms, neutrality
and elementary monotonicity, which we define next. To define neutrality, consider any permutation $\sigma: A \rightarrow A$
of the set of objects. For every preference $P_i$, define $P_i^{\sigma}$ as the preference that satisfies:
$aP_ib$ if and only if $\sigma(a)P^{\sigma}_i\sigma(b)$.
Let $\mathbf{P}^{\sigma}$ be the
preference profile generated by permuting each preference in the preference profile $\mathbf{P}$
by the permutation $\sigma$.
\begin{defn}
A mechanism $Q$ is {\bf neutral} if for every $\mathbf{P}$ and every
permutation $\sigma$,
\begin{align*}
Q_{ia}(\mathbf{P}) &= Q_{i\sigma(a)}(\mathbf{P}^{\sigma})~\qquad~\forall~i \in N,~\forall~a \in A.
\end{align*}
\end{defn}
Neutrality requires that objects be treated symmetrically by the mechanism. Any mechanism which
does not use the ``names" of the objects is neutral -- this includes all priority mechanisms (including the random priority mechanism),
the simultaneous eating algorithms (including the probabilistic serial mechanism) in \citet{Bogo01}.

Our next axiom is elementary monotonicity, an axiom which requires a mild form of monotonicity.
This was introduced in \citet{Majumdar04}.
To define it, we need the notion of ``adjacency" of preferences. We say preferences $P_i$
and $P'_i$ are {\bf adjacent} if there exists a $k \in \{1,\ldots,n-1\}$ such that
$$P_i(k)=P'_i(k+1), P_i(k+1)=P'_i(k),~\textrm{and}~P_i(k')=P'_i(k')~\forall~k' \notin \{k,k+1\}.$$
In other words, $P'_i$ is obtained by swapping consecutively
ranked objects in $P_i$. Here, if $P_i(k)=a$ and $P_i(k+1)=b$, we say that $P'_i$ is an $(a,b)$-swap of $P_i$.
\begin{defn}\label{def:em}
A mechanism $Q$ satisfies {\bf elementary monotonicity} if for every $i \in N$, every $P_{-i} \in \mathcal{P}^{n-1}$,
and every $P_i, P'_i \in \mathcal{P}$ such that $P'_i$ is an $(a,b)$-swap of $P_i$ for some $a,b$, we have
\begin{align}
Q_{ib}(P'_i,P_{-i}) &\ge Q_{ib}(P_i,P_{-i}) \label{eq:em1}\\
Q_{ia}(P'_i,P_{-i}) &\le Q_{ia}(P_i,P_{-i}) \label{eq:em2}
\end{align}
\end{defn}
In other words, as agent $i$ lifts alternative $b$ in ranking by one position by swapping it with $a$ (and keeping
the ranking of every other object the same), elementary monotonicity requires that the share of object $b$
should weakly increase for agent $i$, while share of object $a$ should weakly decrease.
A similar axiom called {\sl swap monotonicity} is used in \citet{Me21}.\footnote{Swap monotonicity
requires the following change in Definition \ref{def:em}: either $Q_i(P'_i,P_{-i})=Q_i(P_i,P_{-i})$ or
Inequalities (\ref{eq:em1}) and (\ref{eq:em2}) hold with strict inequalities.}

It is not difficult to see that elementary monotonicity is a necessary condition for strategy-proofness -- see \citet{Majumdar04}.
As we show later, elementary monotonicity is satisfied by a variety of
mechanisms - including those which are not strategy-proof. However, every neutral mechanism satisfying
elementary monotonicity is U-OBIC.
\begin{theorem}
\label{theo:uniform}
Every neutral mechanism satisfying elementary monotonicity is U-OBIC.
\end{theorem}
\begin{proof}
Fix a neutral mechanism $Q$ satisfying elementary monotonicity.
The proof goes in various steps. \\

\noindent {\sc Step 1.} Pick an agent $i$ and two preferences
$P_i$ and $P'_i$. Pick any $k \in \{1,\ldots,n\}$ and suppose $P_i(k)=a$ and
$P'_i(k)=b$. We show that the interim shares of $a$ and $b$ are same for agent $i$ in preferences $P_i$ and $P'_i$:
$q_{ia}(P_i) = q_{ib}(P'_i)$.
This is a consequence of uniform prior and neutrality. To see this, let $P'_i = P_i^{\sigma}$
for some permutation $\sigma$ of objects in $A$. Then, $b = \sigma(a)$ and hence, for every
$P_{-i}$, we have
\begin{align*}
Q_{ia}(P_i,P_{-i}) &= Q_{i\sigma(a)}(P_i^{\sigma},P_{-i}^{\sigma}) = Q_{ib}(P'_i,P_{-i}^{\sigma}).
\end{align*}
Due to uniform prior and using the above expression,
\begin{align*}
q_{ia}(P_i) &= \frac{1}{(n!)^{n-1}} \sum_{P_{-i}}Q_{ia}(P_i,P_{-i})  = \frac{1}{(n!)^{n-1}} \sum_{P_{-i}}Q_{ib}(P'_i,P_{-i}^{\sigma})
= \frac{1}{(n!)^{n-1}} \sum_{P_{-i}}Q_{ib}(P'_i,P_{-i}) = q_{ib}(P'_i),
\end{align*}
where the third equality follows from the fact that $\{P_{-i}: P_{-i} \in \mathcal{P}^{n-1}\}=\{P_{-i}^{\sigma}: P_{-i} \in \mathcal{P}^{n-1}\}$.

In view of step 1, with some abuse of notation, we write $q_{ik}$ to denote the interim share of the object at rank $k$ in the preference. We call
$q_i$ the interim {\bf rank vector} of agent $i$. \\

\noindent {\sc Step 2.} Pick an agent $i$ and a preference $P_i$. We show that interim shares are non-decreasing with rank:
$q_{ik} \ge q_{i(k+1)}$ for all $k \in \{1,\ldots,n-1\}$. Fix a number $k$ and let $P_i(k)=a$ and $P_i(k+1)=b$.
Then, consider the preference $P'_i$, which is an $(a,b)$-swap of $P_i$. For every $P_{-i}$,
elementary monotonicity implies $Q_{ia}(P_i,P_{-i}) \ge Q_{ia}(P'_i,P_{-i})$. Due to
uniform prior, $q_{ia}(P_i) \ge q_{ia}(P'_i)$. But by Step 1,
\begin{align*}
q_{ik} = q_{ia}(P_i) \ge q_{ia}(P'_i) = q_{i(k+1)}.
\end{align*}

\noindent {\sc Step 3.} We show that $Q$ is OBIC with respect to the uniform prior.
Suppose agent $i$ has preference $P_i$. By Steps 1 and 2, she gets interim rank
vector $(q_{i1},\ldots,q_{in})$ by reporting $P_i$ with $q_{ij} \ge q_{ij+1}$ for all
$j \in \{1,\ldots,n-1\}$. Suppose she reports $P'_i=P^{\sigma}_i$, where $\sigma$ is some
permutation of set of objects. By Steps 1 and 2,
the interim share vector is a permutation of interim rank vector $q_i$ . Using non-decreasingness of this
interim share vector with respect to ranks, we get
$q_i(P_i) \succ_{P_i} q_i(P'_i)$. Hence, $Q$ is OBIC with respect to
uniform prior.
\end{proof}

Theorem \ref{theo:uniform} generalizes, an analogous result in \citet{Majumdar04}, who
consider the voting problem and {\em only} deterministic mechanisms. They arrive at the
same conclusion as Theorem \ref{theo:uniform} in their model. Theorem \ref{theo:uniform} shows that their
result holds even in the {\em random} assignment problem.

\subsection{Probabilistic serial mechanism and U-OBIC}

\citet{Bogo01} define a family of mechanisms, which they call
the {\bf simultaneous eating algorithms} (SEA). Though the SEAs are
not strategy-proof, they satisfy compelling efficiency
and fairness properties, which we discuss in Section \ref{sec:LROBIC}.
We informally introduce the SEAs -- for a formal discussion, see \citet{Bogo01}.

Each SEA is defined by a (possibly time-varying) {\em eating speed} function for each agent. At every preference profile,
agents simultaneously start ``eating" their favorite objects at a rate equal to their eating speed.
Once an object is completely eaten (i.e., the entire share of $1$ is consumed), the amount eaten by each
agent is the share of that agent of that object. Once an object completely eaten, agents go to their
next preferred object and so on.

If the eating speed of each agent is the same, then the simultaneous eating algorithm is
{\em anonymous}. \citet{Bogo01} call the unique anonymous SEA, the {\bf probabilistic serial} mechanism.
\footnote{For axiomatic characterization of the PS mechanism, see \citet{Bogo12} and
\citet{Hashimoto14}.}

\begin{cor}
Every simultaneous eating algorithm is U-OBIC.
\end{cor}
\begin{proof}
Clearly, the SEAs are neutral since eating speeds do not depend on the objects.
The SEAs also satisfy elementary monotonicity: Theorem 3 in \citet{Cho18} and Theorem 1 in \citet{Me21}.
Hence, by Theorem \ref{theo:uniform}, we are done.
\end{proof}

\section{Locally robust OBIC}
\label{sec:LROBIC}

While the uniform prior is an important prior in decision theory,
it is natural to ask if Theorem \ref{theo:uniform} extends to other
``generic" priors. Though we do not have a full answer to this question,
we have been able to answer this question in negative under a natural
robustness requirement. Our robustness requirement is {\em local}.
Take any (independent and identical) prior $\mu$, and let $\mu'$ be any (independent and identical) prior in
the $\epsilon$-radius ball around $\mu$ (where $\epsilon > 0$), i.e.,
$||\mu(P) - \mu'(P)|| < \epsilon$ for all $P \in \mathcal{P}$. In this case, we
write $\mu' \in B_{\epsilon}(\mu)$. Our local robustness requirement is the following.

\begin{defn}
A mechanism $Q$ is {\bf locally robust OBIC (LROBIC)} with respect to a
prior $\mu$ if there exists an $\epsilon > 0$ such that
for every prior $\mu' \in B_{\epsilon}(\mu)$, $Q$ is OBIC with
respect to $\mu'$.
\end{defn}

It is well known that Bayesian incentive compatibility with respect to all
priors lead to strategy-proofness~\citep{Ledyard78}. Here, we require
OBIC with respect to all independent and identical priors in the $\epsilon$-neighborhood of
an independent and identical prior. \citet{Bhargava15} study a version of LROBIC with
respect to {\em uniform} prior but their robustness also allows the mechanism
to be OBIC with respect to correlated priors. They show that
a large class of voting rules satisfy their notion of LROBIC.
We show that in the random assignment model, LROBIC with respect to {\em any} independent and identical prior
has a very different implication.

\begin{theorem}
\label{theo:lobic}
A mechanism is LROBIC with respect to a prior and satisfies elementary monotonicity if and only if it is strategy-proof.
\end{theorem}

The proof builds on some earlier results. Before giving the proof, we define some notions and preliminary results.
We first decompose OBIC into three conditions.
This decomposition is similar to the decomposition of strategy-proofness in \citet{Me21} --
there are some minor differences in axioms and we look at interim share vectors whereas
they look at ex-post share vectors.

Our decomposition of OBIC uses the following three axioms.
\begin{defn}
A mechanism $Q$ satisfies {\bf interim elementary monotonicity} if for every $i \in N$ and every $P_i,P'_i$
such that $P'_i$ is an $(a,b)$-swap of $P_i$, we have
\begin{align*}
q_{ib}(P'_i) \ge q_{ib}(P_i) \\
q_{ia}(P'_i) \le q_{ia}(P_i).
\end{align*}
\end{defn}

Give a preference $P_i$ of agent $i$ and an object $a \in A$, define $U(a,P_i):=\{x \in A: x~P_i~a\}$
and $L(a,P_i):=\{x \in A: a~P_i~x\}$.

\begin{defn}
A mechanism $Q$ satisfies {\bf interim upper invariance} if for every $i \in N$ and every $P_i,P'_i$
such that $P'_i$ is an $(a,b)$-swap of $P_i$, and for every $x \in U(a,P_i)$, we have
\begin{align*}
q_{ix}(P'_i) = q_{ix}(P_i).
\end{align*}
\end{defn}

\begin{defn}
A mechanism $Q$ satisfies {\bf interim lower invariance} if for every $i \in N$ and every $P_i,P'_i$
such that $P'_i$ is an $(a,b)$-swap of $P_i$, and for every $x \in L(b,P_i)$, we have
\begin{align*}
q_{ix}(P'_i) = q_{ix}(P_i).
\end{align*}
\end{defn}

The following proposition characterizes OBIC using these axioms.

\begin{prop}\label{prop:char}
A mechanism $Q$ is OBIC with respect to a prior if and only if it satisfies interim elementary monotonicity,
interim upper invariance, and interim lower invariance.
\end{prop}

Since the proof of Proposition \ref{prop:char} is similar to the characteration of
strategy-proofness in \citet{Me21}, we skip its proof.\footnote{As we discussed, swap monotonicity discussed in \citet{Me21}
is slightly different than elementary monotonicity (neither conditions imply the other), but in the presence of the other two axioms in \citet{Me21} (upper invariance and lower invariance),
they are equivalent.}
The proof of Theorem \ref{theo:lobic} is based on Proposition \ref{prop:char} and the following lemma.
\begin{lemma}\label{lem:lr}
Suppose $Q$ mechanism is LROBIC with respect to a prior. Then, for every $i$,
for every $P_{-i}$, for every $P_i$ and $P'_i$ such that $P'_i$ is an $(a,b)$ swap of $P_i$, we have
\begin{align*}
Q_{ic}(P_i,P_{-i}) &= Q_{ic}(P'_i,P_{-i})~\qquad~\forall~c \notin \{a,b\}.
\end{align*}
\end{lemma}
\begin{proof}
Pick an agent $i \in N$
and $P_i,P'_i \in \mathcal{P}$ such that
$P'_i$ is an $(a,b)$-swap of $P_i$. Fix some $P_{-i}$. By Proposition \ref{prop:char}, $Q$
satisfies interim upper invariance
and interim lower invariance. Hence, we know that for all $c \notin \{a,b\}$,
we get
\begin{align}\label{eq:obicvar}
\sum_{P_{-i}}\mu(P_{-i}) \Big[Q_{ic}(P_i,P_{-i}) - Q_{ic}(P'_i,P_{-i})\Big] = 0.
\end{align}
Since $\mu$ is a probability distribution over $\mathcal{P}$, we can
treat it as a vector in $\mathbb{R}^{n!-1}$. Using
$\mu(P_{-i}) \equiv \times_{j \ne i}\mu(P_j)$, we note that
the LHS of the Equation (\ref{eq:obicvar}) is a polynomial function of
$\{\mu(P)\}_{P \in \mathcal{P}}$. The equation describes
the zero set of this polynomial function. For non-zero polynomials, the set of zeros has measure zero~\citep{Caron05}, i.e.,
the set of $\mu$ satisfying Equation (\ref{eq:obicvar}) has measure zero.\footnote{Note that
we are not characterizing the set of $\mu$ for which (\ref{eq:obicvar}) has a solution. Our claim
is only about the measure of the set of solutions.}
Hence, given any prior $\mu^*$ and $\epsilon > 0$, if
Equation \ref{eq:obicvar} has to hold for {\em all} $\mu \in B_{\epsilon}(\mu^*)$ (which
has non-zero measure),
then $Q_{ic}(P_i,P_{-i})=Q_{ic}(P'_i,P_{-i})$ for all $c \notin \{a,b\}$.
\end{proof}

We now complete the proof of Theorem \ref{theo:lobic}. \\

\noindent {\sl Proof of Theorem \ref{theo:lobic}} \\

\begin{proof}
Every strategy-proof mechanism is OBIC with respect to any prior. A strategy-proof mechanism satisfies elementary monotonicity.
So, we now focus on the other direction of the proof.
Let $Q$ be an LROBIC mechanism with respect to a prior $\mu$. Suppose $Q$ satisfies elementary monotonicity.

By Lemma \ref{lem:lr}, any LROBIC mechanism $Q$, $Q$ satisfies {\em ex-post} versions of interim lower invariance and
interim upper invariance. \citet{Me21} refer to these properties as
{\em upper invariance} and {\em lower invariance} (see also \citet{Cho18}).
They show that upper invariance, lower invariance, and elementary monotonicity are
equivalent to strategy-proofness.
By the assumption of the theorem, $Q$ satisfies elementary monotonicity. Hence, it is strategy-proof.
\end{proof}

We now explore the compatibility of LROBIC and ordinal efficiency.

\begin{defn}
A mechansim $Q$ is {\bf ordinally efficient} if at every preference profile
$\mathbf{P}$ there exists no assignment $L$ such that
\begin{align*}
L_i \succ_{P_i} Q_i(\mathbf{P})~\qquad~\forall~i \in N,
\end{align*}
with $L_i \ne Q_i(\mathbf{P})$ for some $i$.
\end{defn}

\citet{Bogo01} show that every ordinally efficient mechanism is ex-post efficient but
the converse is not true if $n \ge 4$. In fact, for $n \ge 4$, strategy-proofness is
incompatible with ordinally efficiency along with the following weak fairness criterion.

\begin{defn}
A mechanism $Q$ satisfies {\bf equal treatment of equals} if at every preference
profile $\mathbf{P}$ and for every $i,j \in N$, we have
\begin{align*}
\Big[ P_i = P_j \Big] \Rightarrow \Big[ Q_i(\mathbf{P})=Q_j(\mathbf{P})\Big]
\end{align*}
\end{defn}

Due to Theorem \ref{theo:lobic}, we can strengthen the impossibility results in \citet{Bogo01}
 and \citet{Mennle17} as follows.

\begin{cor}\label{cor:imp}
Suppose $n \ge 4$. Then, there is no locally robust OBIC and ordinally efficient
mechanism satisfying equal treatment of equals.
\end{cor}
\begin{proof}
By Lemma \ref{lem:lr}, a locally robust OBIC mechanism satisfies {\em ex-post}
versions of upper invariance and lower invariance.
\citet{Mennle17} show that the proof in \citet{Bogo01} can be adapted by replacing
strategy-proofness with ex-post versions of upper invariance and lower invariance. Hence, these two properties are incompatible with ordinal efficiency
and equal treatment of equals for $n \ge 4$, and we are done.
\end{proof}

Note that Corollary \ref{cor:imp} does not use elementary monotonicity, and hence, cannot be directly inferred
from Theorem \ref{theo:lobic}.

%\bibliographystyle{ecta}
%\bibliography{order}

\begin{thebibliography}{32}
\newcommand{\enquote}[1]{``#1''}
\expandafter\ifx\csname natexlab\endcsname\relax\def\natexlab#1{#1}\fi

\bibitem[\protect\citeauthoryear{Abebe, Cole, Gkatzelis, and Hartline}{Abebe
  et~al.}{2020}]{Ab20}
\textsc{Abebe, R., R.~Cole, V.~Gkatzelis, and J.~D. Hartline} (2020):
  \enquote{A truthful cardinal mechanism for one-sided matching,} in
  \emph{Proceedings of the Fourteenth Annual ACM-SIAM Symposium on Discrete
  Algorithms}, 2096--2113.

\bibitem[\protect\citeauthoryear{Aziz, Gaspers, Mackenzie, Mattei, Narodytska,
  and Walsh}{Aziz et~al.}{2015}]{Az15}
\textsc{Aziz, H., S.~Gaspers, S.~Mackenzie, N.~Mattei, N.~Narodytska, and
  T.~Walsh} (2015): \enquote{Equilibria under the probabilistic serial rule,}
  \emph{arXiv preprint arXiv:1502.04888}.

\bibitem[\protect\citeauthoryear{Aziz, Gaspers, Mattei, Narodytska, and
  Walsh}{Aziz et~al.}{2014}]{Az14}
\textsc{Aziz, H., S.~Gaspers, N.~Mattei, N.~Narodytska, and T.~Walsh} (2014):
  \enquote{Strategic aspects of the probabilistic serial rule for the
  allocation of goods,} \emph{arXiv preprint arXiv:1401.6523}.

\bibitem[\protect\citeauthoryear{Balbuzanov}{Balbuzanov}{2016}]{Balbu16}
\textsc{Balbuzanov, I.} (2016): \enquote{Convex strategyproofness with an
  application to the probabilistic serial mechanism,} \emph{Social Choice and
  Welfare}, 46, 511--520.

\bibitem[\protect\citeauthoryear{Bhargava, Majumdar, and Sen}{Bhargava
  et~al.}{2015}]{Bhargava15}
\textsc{Bhargava, M., D.~Majumdar, and A.~Sen} (2015):
  \enquote{Incentive-compatible voting rules with positively correlated
  beliefs,} \emph{Theoretical Economics}, 10, 867--885.

\bibitem[\protect\citeauthoryear{Bogomolnaia and Heo}{Bogomolnaia and
  Heo}{2012}]{Bogo12}
\textsc{Bogomolnaia, A. and E.~J. Heo} (2012): \enquote{Probabilistic
  assignment of objects: Characterizing the serial rule,} \emph{Journal of
  Economic Theory}, 147, 2072--2082.

\bibitem[\protect\citeauthoryear{Bogomolnaia and Moulin}{Bogomolnaia and
  Moulin}{2001}]{Bogo01}
\textsc{Bogomolnaia, A. and H.~Moulin} (2001): \enquote{A new solution to the
  random assignment problem,} \emph{Journal of Economic theory}, 100, 295--328.

\bibitem[\protect\citeauthoryear{Bogomolnaia and Moulin}{Bogomolnaia and
  Moulin}{2002}]{Bogo02}
---\hspace{-.1pt}---\hspace{-.1pt}--- (2002): \enquote{A simple random
  assignment problem with a unique solution,} \emph{Economic Theory}, 19,
  623--636.

\bibitem[\protect\citeauthoryear{Caron and Traynor}{Caron and
  Traynor}{2005}]{Caron05}
\textsc{Caron, R. and T.~Traynor} (2005): \enquote{The zero set of a
  polynomial,} University of Windsor, Technical Report.

\bibitem[\protect\citeauthoryear{Che and Kojima}{Che and Kojima}{2010}]{Che10}
\textsc{Che, Y.-K. and F.~Kojima} (2010): \enquote{Asymptotic equivalence of
  probabilistic serial and random priority mechanisms,} \emph{Econometrica},
  78, 1625--1672.

\bibitem[\protect\citeauthoryear{Cho}{Cho}{2018}]{Cho18}
\textsc{Cho, W.~J.} (2018): \enquote{Probabilistic assignment: an extension
  approach,} \emph{Social Choice and Welfare}, 51, 137--162.

\bibitem[\protect\citeauthoryear{d'Aspremont and Peleg}{d'Aspremont and
  Peleg}{1988}]{dAsp88}
\textsc{d'Aspremont, C. and B.~Peleg} (1988): \enquote{Ordinal Bayesian
  incentive compatible representations of committees,} \emph{Social Choice and
  Welfare}, 5, 261--279.

\bibitem[\protect\citeauthoryear{Ehlers and Mass{\'o}}{Ehlers and
  Mass{\'o}}{2007}]{Ehlers07}
\textsc{Ehlers, L. and J.~Mass{\'o}} (2007): \enquote{Incomplete information
  and singleton cores in matching markets,} \emph{Journal of Economic Theory},
  136, 587--600.

\bibitem[\protect\citeauthoryear{Gibbard}{Gibbard}{1973}]{Gibbard73}
\textsc{Gibbard, A.} (1973): \enquote{Manipulation of voting schemes: a general
  result,} \emph{Econometrica}, 41, 587--601.

\bibitem[\protect\citeauthoryear{Gibbard}{Gibbard}{1977}]{Gibbard77}
---\hspace{-.1pt}---\hspace{-.1pt}--- (1977): \enquote{Manipulation of schemes
  that mix voting with chance,} \emph{Econometrica}, 45, 665--681.

\bibitem[\protect\citeauthoryear{Hashimoto, Hirata, Kesten, Kurino, and
  {\"U}nver}{Hashimoto et~al.}{2014}]{Hashimoto14}
\textsc{Hashimoto, T., D.~Hirata, O.~Kesten, M.~Kurino, and M.~U. {\"U}nver}
  (2014): \enquote{Two axiomatic approaches to the probabilistic serial
  mechanism,} \emph{Theoretical Economics}, 9, 253--277.

\bibitem[\protect\citeauthoryear{Hong and Kim}{Hong and Kim}{2018}]{Kim18}
\textsc{Hong, M. and S.~Kim} (2018): \enquote{Unanimity and local incentive
  compatibility,} Working paper, Yonsei University.

\bibitem[\protect\citeauthoryear{Katta and Sethuraman}{Katta and
  Sethuraman}{2006}]{Katta06}
\textsc{Katta, A.-K. and J.~Sethuraman} (2006): \enquote{A solution to the
  random assignment problem on the full preference domain,} \emph{Journal of
  Economic theory}, 131, 231--250.

\bibitem[\protect\citeauthoryear{Kojima and Manea}{Kojima and
  Manea}{2010}]{Kojima10}
\textsc{Kojima, F. and M.~Manea} (2010): \enquote{Incentives in the
  probabilistic serial mechanism,} \emph{Journal of Economic Theory}, 145,
  106--123.

\bibitem[\protect\citeauthoryear{Ledyard}{Ledyard}{1978}]{Ledyard78}
\textsc{Ledyard, J.~O.} (1978): \enquote{Incentive compatibility and incomplete
  information,} \emph{Journal of Economic Theory}, 18, 171 -- 189.

\bibitem[\protect\citeauthoryear{Liu}{Liu}{2019}]{Liu19}
\textsc{Liu, P.} (2019): \enquote{Random assignments on sequentially
  dichotomous domains,} \emph{Games and Economic Behavior}.

\bibitem[\protect\citeauthoryear{Liu and Zeng}{Liu and Zeng}{2019}]{Liu19b}
\textsc{Liu, P. and H.~Zeng} (2019): \enquote{Random assignments on preference
  domains with a tier structure,} \emph{Journal of Mathematical Economics}, 84,
  176--194.

\bibitem[\protect\citeauthoryear{Majumdar and Sen}{Majumdar and
  Sen}{2004}]{Majumdar04}
\textsc{Majumdar, D. and A.~Sen} (2004): \enquote{Ordinally Bayesian incentive
  compatible voting rules,} \emph{Econometrica}, 72, 523--540.

\bibitem[\protect\citeauthoryear{Mennle and Seuken}{Mennle and
  Seuken}{2017}]{Mennle17}
\textsc{Mennle, T. and S.~Seuken} (2017): \enquote{Two New Impossibility
  Results for the Random Assignment Problem,} Working paper, University of
  Zurich, https://www.ifi.uzh.ch/ce/publications/IPS.pdf.

\bibitem[\protect\citeauthoryear{Mennle and Seuken}{Mennle and
  Seuken}{2021}]{Me21}
---\hspace{-.1pt}---\hspace{-.1pt}--- (2021): \enquote{Partial
  strategyproofness: Relaxing strategyproofness for the random assignment
  problem,} \emph{Journal of Economic Theory}, 191, 105144.

\bibitem[\protect\citeauthoryear{Miralles}{Miralles}{2012}]{Mi12}
\textsc{Miralles, A.} (2012): \enquote{Cardinal Bayesian allocation mechanisms
  without transfers,} \emph{Journal of Economic Theory}, 147, 179--206.

\bibitem[\protect\citeauthoryear{Mishra}{Mishra}{2016}]{Mishra16}
\textsc{Mishra, D.} (2016): \enquote{Ordinal Bayesian incentive compatibility
  in restricted domains,} \emph{Journal of Economic Theory}, 163, 925--954.

\bibitem[\protect\citeauthoryear{P{\'a}pai}{P{\'a}pai}{2000}]{Papai00}
\textsc{P{\'a}pai, S.} (2000): \enquote{Strategyproof assignment by
  hierarchical exchange,} \emph{Econometrica}, 68, 1403--1433.

\bibitem[\protect\citeauthoryear{Pycia and {\"U}nver}{Pycia and
  {\"U}nver}{2017}]{Pycia17}
\textsc{Pycia, M. and M.~U. {\"U}nver} (2017): \enquote{Incentive compatible
  allocation and exchange of discrete resources,} \emph{Theoretical Economics},
  12, 287--329.

\bibitem[\protect\citeauthoryear{Satterthwaite}{Satterthwaite}{1975}]{Satt75}
\textsc{Satterthwaite, M.~A.} (1975): \enquote{Strategy-proofness and Arrow's
  conditions: Existence and correspondence theorems for voting procedures and
  social welfare functions,} \emph{Journal of economic theory}, 10, 187--217.

\bibitem[\protect\citeauthoryear{Svensson}{Svensson}{1999}]{Svensson99}
\textsc{Svensson, L.-G.} (1999): \enquote{Strategy-proof allocation of
  indivisible goods,} \emph{Social Choice and Welfare}, 16, 557--567.

\bibitem[\protect\citeauthoryear{Wilson}{Wilson}{1987}]{Wilson87}
\textsc{Wilson, R.} (1987): \enquote{Game theoretic approaches to trading
  processess,} in \emph{Advances in Economic Theory: Fifth World Congress of
  Econometric Society}, ed. by T.~Bewley, Cambridge, UK: Cambridge University
  Press, 201--213.

\end{thebibliography}
\end{document}